\documentclass[runningheads,unicode]{llncs}
\usepackage[crop]{llncsconf}
\usepackage[T1]{fontenc}
\usepackage{graphicx}
\usepackage[T1]{fontenc}

\usepackage{amsmath, amssymb}
\usepackage{url}
\usepackage[linesnumbered, ruled, vlined]{algorithm2e}
\usepackage{xcolor}
\usepackage{tikz}
\usepackage{float} 
\usetikzlibrary{arrows.meta}
\usepackage{color} 
\usepackage{cite}
\usepackage{tikz}
\usetikzlibrary{positioning, shapes, arrows.meta}
\usepackage{cleveref}
\usetikzlibrary{shapes.geometric}
\usepackage{comment}

\newcommand{\tree}{{\mathtt{tree}}}
\newcommand{\inc}{{\mathtt{inc}}}

\newcommand{\paramstyle}[1]{\text{\sf #1}}
\newcommand{\tw}{\paramstyle{tw}}

\newcommand{\cvd}{\paramstyle{cvd}}

\newenvironment{claimproof}[1][Proof of the Claim]{%
  \par\noindent\textit{#1.} \ }%
  {\hfill\ensuremath{\blacksquare}\par}
\usepackage[mathlines]{lineno}
\newcommand*\patchAmsMathEnvironmentForLineno[1]{
  \expandafter\let\csname old#1\expandafter\endcsname\csname #1\endcsname
  \expandafter\let\csname oldend#1\expandafter\endcsname\csname end#1\endcsname
  \renewenvironment{#1}
     {\linenomath\csname old#1\endcsname}
     {\csname oldend#1\endcsname\endlinenomath}}
\newcommand*\patchBothAmsMathEnvironmentsForLineno[1]{
  \patchAmsMathEnvironmentForLineno{#1}
  \patchAmsMathEnvironmentForLineno{#1*}}
\AtBeginDocument{
\patchBothAmsMathEnvironmentsForLineno{equation}
\patchBothAmsMathEnvironmentsForLineno{align}
\patchBothAmsMathEnvironmentsForLineno{flalign}
\patchBothAmsMathEnvironmentsForLineno{alignat}
\patchBothAmsMathEnvironmentsForLineno{gather}
\patchBothAmsMathEnvironmentsForLineno{multline}
}

\begin{document}
\title{Finding a HIST: Chordality, Structural Parameters, and Diameter\thanks{This work is partially supported by 
JP21K17707, 
JP22H00513, 
JP23H04388, 
JP24K02898, 
JP25K03077, 
and
JST, CRONOS, JPMJCS24K2.
}
}
%
%
\author{
Tesshu Hanaka \and  
Hironori Kiya \and
Hirotaka Ono
}
\authorrunning{T. Hanaka et al.}
%
\institute{Kyushu University, Fukuoka 819-0395, Japan \email{hanaka@inf.kyushu-u.ac.jp}\and 
Nagoya University, Nagoya 464-8601, Japan \email{ono@nagoya-u.jp}\and 
Osaka Metropolitan University, Osaka 599-8531 
\email{kiya@omu.ac.jp}}
\maketitle              
%
\begin{abstract}
A \emph{homeomorphically irreducible spanning tree} (HIST) is a spanning tree with no degree-2 vertices, serving as a structurally minimal backbone of a graph. While the existence of HISTs has been widely studied from a structural perspective, the algorithmic complexity of finding them remains less understood. In this paper, we provide a comprehensive investigation of the HIST problem from both structural and algorithmic viewpoints.  
We present a simple characterization that precisely describes which chordal graphs of diameter at most~3 admit a HIST, leading to a polynomial-time decision procedure for this class. In contrast, we show that the problem is NP-complete for strongly chordal graphs of diameter~4. From the perspective of parameterized complexity, we establish that the HIST problem is W[1]-hard when parameterized by clique-width, indicating that the problem is unlikely to be efficiently solvable in general dense graphs. On the other hand, we present fixed-parameter tractable (FPT) algorithms when parameterized by treewidth, modular-width, or cluster vertex deletion number. Specifically, we develop an $O^*(4^{k})$-time algorithm parameterized by modular-width~$k$, and an FPT algorithm parameterized by the cluster vertex deletion number based on kernelization techniques that bound clique sizes while preserving the existence of a HIST. These results together provide a clearer understanding of the structural and computational boundaries of the HIST problem.
\keywords{spanning tree \and HIST \and chordal graph \and diameter \and modular width.}
\end{abstract}
\section{Introduction}
Let \( G = (V, E) \) be a connected graph. A \emph{homeomorphically irreducible spanning tree} (HIST) of \( G \) is a spanning tree of \( G \) that contains no vertex with degree two~\cite{hill1974graphs}. 
Because a HIST has no degree-2 vertex, it cannot be obtained by subdividing the edges of a smaller tree; hence, it is called ``homeomorphically irreducible.''

This exclusion imposes a particular structural constraint compared to general spanning trees, often resulting in a more branching or less linear configuration. In trees, degree-2 vertices typically function as internal points along paths. By eliminating them, the definition enforces that all internal vertices (i.e., branches) must have a degree of at least 3. This structural property makes HISTs particularly relevant in contexts where long chains of vertices are undesirable or where higher branching connectivity is preferred. Such trees or trees with a few degree-2 vertices have applications in the design of efficient communication networks with minimal redundancy~\cite{ZHALECHIAN201820,DBLP:journals/integration/CongHKM96}. Beyond their practical relevance, HISTs are also of significant theoretical interest, as they play a key role in major conjectures in graph theory, such as the 3-decomposition conjecture~\cite{HoffmannOstenhof2011,Bachtler2023}. For these reasons, HISTs have been extensively studied from both applied and theoretical perspectives.

HISTs have been studied extensively in graph theory, particularly with respect to their existence. For example, the paper by Albertson et al.~\cite{DBLP:journals/jgt/AlbertsonBHT90} shows that if a graph with $n$ vertices has minimum degree at least $\sqrt{n}$, then it admits a HIST, and that this bound is nearly tight. 
The paper also demonstrates that some graphs with diameter two do not admit a HIST, but proves that every such graph has a spanning tree with at most three degree-2 vertices. 
Later, Shan and Tsuchiya~\cite{DBLP:journals/jgt/ShanT23} completely characterized diameter-2 graphs that contain a HIST. 
Albertson et al.~\cite{DBLP:journals/jgt/AlbertsonBHT90} also resolved a conjecture by Hill~\cite{hill1974graphs} by confirming that every triangulation of the plane has a HIST. They extended this to near-triangulations, showing that the structure of planar embeddings can guarantee the existence of a HIST under certain face and edge conditions. 
Importantly, in 2013, Chen and Shan~\cite{DBLP:journals/jct/ChenS13} resolved an open problem that had remained unsolved for more than 20 years~\cite{Malkevitch1979,DBLP:journals/jgt/AlbertsonBHT90}, proving that any graph in which every edge is contained in at least two triangles contains a HIST.

However, from a practical standpoint, the existence of a HIST is not the only aspect of interest. In many applications, constructing a HIST, rather than merely proving its existence, is equally or even more critical. Despite this, the algorithmic task of constructing a HIST has received relatively little attention in the literature. 
It is perhaps not surprising that deciding whether a given graph admits a HIST is NP-complete~\cite{DBLP:journals/jgt/AlbertsonBHT90}, even for planar graphs with maximum degree at most 3~\cite{DBLP:journals/dm/Douglas92}. 
However, as far as the authors know, few additional results are known beyond these.

In this paper, we address algorithms and computational complexity for finding HISTs from two complementary perspectives.
The first perspective investigates the relationship between the existence of HISTs and structural properties of graphs, such as graph diameter and triangulated structures, which have been central themes in previous studies.
Regarding graph diameter, small diameters, such as 3 or 4, are of particular interest based on the series of prior works. On the side of triangulated structures, we focus on \emph{chordal graphs}, also known as \emph{triangulated graphs}. 
Chordal graphs, characterized by the presence of chords in every cycle, form a graph class where many problems that are otherwise computationally hard can be solved efficiently. In particular, their strong relationship with tree decompositions plays a key role in algorithm design. Furthermore, the class of chordal graphs includes several well-known subclasses, such as split graphs, block graphs, and interval graphs, which have been widely studied. Based on these, we examine the existence of HISTs and the computational complexity of finding them in chordal graphs with small diameters.

The second perspective focuses on algorithmic aspects, particularly on the design of efficient algorithms for constructing HISTs in general graphs.  
To circumvent the NP-hardness of the problem, two primary approaches are commonly considered: the design of approximation algorithms or parameterized algorithms.  
Since the HIST problem is also a decision problem, it is natural to focus on the latter, namely, the design of parameterized algorithms.  
In this study, we investigate whether the HIST problem is fixed-parameter tractable (FPT) with respect to structural graph parameters such as treewidth and modular-width.  
The detailed results of this study are presented in \Cref{sec:contribution}.

\subsection{Our Contribution}\label{sec:contribution}
\paragraph*{Chordal graphs with a small diameter}
As seen in the previous section, prior studies on the existence of HISTs have primarily focused on the role of triangle structures and graph diameter. Shan and Tsuchiya~\cite{DBLP:journals/jgt/ShanT23} provide a complete characterization of graphs with diameter 2 that admit a HIST. Based on this result, we first confirm that such graphs can be recognized in polynomial time.

We then investigate split graphs and block-split graphs, which are representative examples of chordal graphs with diameter 3. In this paper, we first characterize block-split graphs that admit a HIST (\Cref{thm:blocksplit}) and then extend this characterization to split graphs (\Cref{thm:split}) and chordal graphs with diameter $3$ (\Cref{thm:d3}) that admits a HIST. Using these characterizations, we can find a HIST of a given chordal graph with diameter $3$ in polynomial time. On the other hand, we prove that deciding whether a strongly chordal graph with diameter 4 admits a HIST is NP-complete, which provides a sharp boundary.
The results can be summarized as follows:
\begin{itemize}
    \item For graphs with diameter at most 2, the existence of a HIST can be decided in polynomial time [\Cref{thm:diameter2}].
    \item In contrast, deciding whether a graph with diameter 4 admits a HIST is NP-complete, even when restricted to strongly chordal graphs [\Cref{thm:NP:diameter}].
    \item For chordal graphs with diameter 3, we give a simple characterization of graphs that admit a HIST, which enables us to determine whether it has a HIST in polynomial time [\Cref{thm:d3}]. However, the computational complexity of deciding whether a given graph with diameter 3 contains a HIST remains open. 
\end{itemize}


\paragraph*{Algorithms for structural graph parameters}
As a basic result, we first present a \( 4^n n^{O(1)} \)-time algorithm.  
We then investigate its parameterized complexity. Since the reduction in the proof of \Cref{thm:NP:diameter} implies that the problem is W[1]-hard when parameterized by clique-width, we turn our attention to more specific parameters: treewidth, modular-width, and cluster-vertex-deletion number.

These parameters provide upper bounds on clique-width and characterize more restricted graph classes (see e.g., ~\cite{tran2022expanding}). Moreover, they capture distinct structural properties: treewidth reflects global tree-likeness, modular-width measures recursive modularity, and cluster-vertex-deletion number quantifies proximity to clustered structures. By studying the problem with respect to these parameters, we aim to clarify the structural conditions under which the problem becomes tractable, providing a more fine-grained perspective beyond clique-width.

\Cref{fig:graph-parameters} summarizes the current landscape of parameterized complexity with respect to these structural graph parameters.  
These results collectively help clarify the structural boundaries that govern the computational complexity of finding a HIST.


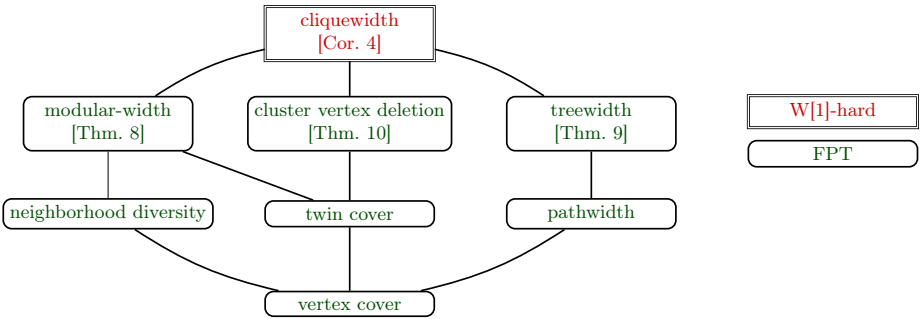
\begin{figure}[ht]
\centering
\scalebox{0.8}{ 
\begin{tikzpicture}[
  every node/.style={draw, rectangle, rounded corners, align=center, minimum width=2.8cm, font=\small},
  Whard/.style={text=red!80!black,draw, rectangle,  double, sharp corners},
  XP/.style={text=blue, draw, rectangle, sharp corners},
  FPTAS/.style={text=orange,dash pattern=on 3pt off 2pt},
  FPT/.style={text=green!30!black,thick},
  Open/.style={text=black,dash pattern=on 3pt off 2pt,sharp corners,thick},
]

\node[FPT] (cvd) at (0,3) {cluster vertex deletion\\ \text{[Thm. \ref{thm:cvd}]}};
\node[FPT] (tc) at (0,1.5) {twin cover};

\node[FPT] (mw) at (-4,3) {modular-width\\ \text{[Thm. \ref{thm:mw}]}};
\node[FPT] (nd) at (-4,1.5) {neighborhood diversity};

\node[Whard] (cw) at (0,4.5) {cliquewidth\\ \text{[Cor. \ref{cor:NP:other}]}};
\node[FPT] (tw) at (4,3) {treewidth\\ \text{[Thm. \ref{thm:treewidth}]}};
\node[FPT] (pw) at (4,1.5) {pathwidth};
\node[FPT] (vc) at (0,0) {vertex cover};

\draw [thick] (cw) to[bend left=10] (tw);
\draw (mw) -- (nd);


\draw [thick] (cw) -- (cvd);
\draw [thick] (cvd) -- (tc);
\draw [thick] (mw) -- (tc);

\draw [thick] (tw) -- (pw);
\draw [thick] (tc) -- (vc);

\draw [thick] (pw) to[bend left=10] (vc);

\draw [thick] (cw)  to[bend right=10]   (mw);

\draw [thick] (nd) to[bend right=10] (vc);

\node[Whard] at (8,3.2) {\textcolor{red!80!black}{W[1]-hard}};
\node[FPT] at (8,2.5) {\textcolor{green!30!black}{FPT}};

\end{tikzpicture}
}
\caption{Parameterized complexity of HIST with respect to structural graph parameters.
The connection between two parameters means that the upper parameter $p$ is bounded by some
computable function $f(\cdot)$ of the lower parameter $q$, i.e., $p \le f (q)$.
The double and rounded rectangles indicate that the problem is W[1]-hard and fixed-parameter tractable, respectively.
}
\label{fig:graph-parameters}
\end{figure}

\subsection{HIST: history, related work, and related notions}
Research on HISTs began with Hill in 1974~\cite{hill1974graphs}, who conjectured that every triangulation of the sphere with minimum degree at least 4 contains a HIST. This initial conjecture led to early studies focusing on planar graphs.  
In 1979, Malkevitch strengthened the conjecture by extending it to near-triangulations~\cite{Malkevitch1979}.  
In 1990, Albertson, Berman, Hutchinson, and Thomassen not only confirmed Malkevitch's conjecture but also raised a broader question: whether triangulations on any surface admit a HIST~\cite{DBLP:journals/jgt/AlbertsonBHT90}. This was later formalized as a conjecture by Archdeacon~\cite{Archdeacon_2009}. The same article also suggests that the conjecture may hold for a broader class of graphs, specifically those in which every edge belongs to at least two triangles. 
These conjectures are partially resolved by Chen, Ren, and Shan~\cite{DBLP:journals/cpc/ChenRS12}, who proved that every connected and locally connected graph with minimum degree greater than 3 contains a HIST. Finally, Chen and Shan~\cite{DBLP:journals/jct/ChenS13} resolved the conjecture by Archdeacon that any graph in which every edge is contained in at least two triangles contains a HIST. 

Closely related to HISTs, the 3-decomposition conjecture, formulated by Hoffmann-Ostenhof in 2011, stands as a fundamental open question in graph theory. It asserts that every connected cubic graph admits a decomposition into a spanning tree, a 2-regular subgraph, and a matching~\cite{HoffmannOstenhof2011,Cameron2011}.
The conjecture is intrinsically linked to HISTs: such a decomposition can always be derived from a HIST, and any graph containing a HIST trivially satisfies the conjecture~\cite{HoffmannOstenhof2011}. HISTs, thus, serve as a central concept in addressing this problem, and thus some work focuses on the existence of a HIST of a cubic graph~\cite{DBLP:journals/jgt/Hoffmann-Ostenhof18a}.
The conjecture has been verified for several graph classes, including planar graphs~\cite{DBLP:journals/jgt/Hoffmann-Ostenhof18}, claw-free graphs~\cite{ahanjideh2018hoffmannostenhofsconjectureclawfreecubic}, and graphs with pathwidth at most 4~\cite{DBLP:conf/lagos/BachtlerH23}. These results emphasize the importance of structural graph parameters and the critical role of HISTs as fundamental tools in graph decomposition theory. 


Another notion closely related to HIST is \emph{odd spanning tree}. A spanning tree is called \emph{odd} if every vertex has an odd degree, and a spanning tree of a graph is called an odd spanning tree if it is odd. 
Since an odd spanning tree is a HIST, it is a special case of HISTs. 
Recently, Zheng et al.~\cite{zheng2025odd} have studied \emph{odd spanning trees}, and provided a necessary and sufficient condition that a split graph admits an odd spanning tree.  


\medskip 

The remainder of this paper is organized as follows. Section~2 introduces fundamental definitions and basic results. Section~3 provides structural characterizations of chordal graphs with diameter at most~3, including a simple necessary and sufficient condition for the existence of a HIST. Section~4 presents hardness results, showing that deciding whether a strongly chordal graph with diameter~4 admits a HIST is NP-complete. Section~5 develops an exact exponential-time algorithm for the general case. Section~6 focuses on parameterized algorithms, presenting fixed-parameter tractability results with respect to treewidth, modular-width, and cluster vertex deletion number.

\section{Preliminaries}
\subsection{Notations and terminology}
We assume basic knowledge of graph theory and discrete algorithms. 
All graphs considered in this paper are finite, simple, and undirected.
Let $G = (V(G), E(G)) $ be a simple undirected graph, where $V(G)$ and $E(G)$ respectively denote the set of vertices and the set of edges of $G$. 
We sometimes simply write $V$ and $E$ to refer to $V(G)$ and $E(G)$, respectively. 
A graph $G$ is said to be \emph{connected} if a path exists between every pair of vertices in $V$.
For a graph $G = (V, E)$ and a vertex $v \in V$, the \emph{neighborhood} of $v$, denoted by $N_G(v)$, is the set of vertices adjacent to $v$, that is, $N_G(v) = \{ u \in V \mid \{u, v\} \in E \}$. 
%
%
The \emph{degree} of a vertex $v \in V$ in $G$, denoted by $d_G(v)$, is the number of vertices in $N_G(v)$, i.e., $|N_G(v)|$. Both $N_G(v)$ and $d_G(v)$ may be simply denoted respectively by $N(v)$ and $d(v)$, if the considered graph $G$ is clear from context.  
An edge $\{u, v\} \in E$ is called a \emph{pendant edge} if one of its endpoints has degree one. A vertex of degree one is called a \emph{pendant vertex}, the endpoint of a pendant edge.
\emph{Twins} are a pair of vertices that share the same neighbors, excluding each other. When twins are adjacent, they are called \emph{true twins}; when they are non-adjacent, they are called \emph{false twins}. Formally, vertices \( u \) and \( v \) are twins if \( N(u) \setminus \{v\} = N(v) \setminus \{u\} \). They are \emph{true twins} if they are adjacent, and \emph{false twins} if they are not adjacent and satisfy \( N(u) = N(v) \).
A \emph{path} in $G$ is a sequence of distinct vertices $(v_0, v_1, \ldots, v_k)$ such that $\{v_{i-1}, v_i\} \in E$ for all $i=1,\dots, k$. The length of the path is defined by the number of its edges, i.e., $k$.
The \emph{distance} between two vertices $u$ and $v$ in $G$ is the length of the shortest path between $u$ and $v$.
The \emph{diameter} of $G$ is the maximum distance between any pair of vertices in $V$.
A \emph{spanning tree} of a connected graph $G = (V, E)$ is a subgraph $T = (V, E')$ such that $T$ is a tree and $E' \subseteq E$. 

A graph is called \emph{chordal} if every cycle of length at least four has a \emph{chord}, an edge connecting two non-consecutive vertices in the cycle.
A graph is called \emph{strongly chordal} if it is chordal and every even cycle of length at least six has an odd chord, that is, a chord that divides the cycle into two paths of odd length.
A graph $G = (V, E)$ is called a \emph{split graph} if its vertex set $V$ can be partitioned into two disjoint sets $C$ and $I$, where $C$ induces a clique and $I$ induces an independent set.
A \emph{block-split graph} is a split graph $G = (C, I, E)$ such that each vertex in the independent set $I$ has degree at most one. Every vertex in $I$ is isolated or adjacent to exactly one vertex in $C$. Equivalently, the subgraph induced by $C \cup I$ contains only pendant edges between $I$ and $C$, and all remaining edges lie within the clique $C$.

Given a graph $G = (V, E)$, a \emph{Hamiltonian path} is a path that visits each vertex of $V$ exactly once. 
In particular, we consider the \emph{$s$-$t$ Hamiltonian path}, which is a Hamiltonian path that starts at a designated vertex $s$ and ends at another designated vertex $t$.
The \emph{Hamiltonian path problem} asks whether a given graph contains a Hamiltonian path between two specified vertices.
It is well known that the ($s$-$t$) Hamiltonian path problem is NP-complete in general graphs. Moreover, the hardness persists even when the input graph is restricted to bipartite, chordal, or split graphs.

A set $V’ \subseteq V$ is called a \emph{dominating set} of $G$ if for all $u\in V\setminus V'$ there is a $v\in V’$ such that $u\in N(v)$. If we require, in addition, that the subgraph induced by $V’$ in $G$ be a clique, then the corresponding set is called a \emph{dominating clique}. 

\medskip




Before concluding this section, we see the proposition used throughout the paper. 
\begin{proposition}\label{prop:subgraph}
For a graph $G=(V,E)$, $G$ contains a HIST if there exists a spanning subgraph $G'$ of $G$ such that $G'$ has a HIST.   
\end{proposition}
In the following sections, we repeatedly use the argument, based on \Cref{prop:subgraph}, that to prove $G$ admits a HIST, it suffices to find an appropriate subgraph $G'$ that has a HIST. In some cases, we explicitly mention this step, while in others, we simply present a subgraph with a HIST without further explanation. 

\subsection{Finding a HIST of a graph with diameter 2}
As a basic result, we here see that finding a HIST of a graph with diameter at most two is done in polynomial time. 
The case of diameter one is trivial. Since a graph with diameter one is a complete graph, it is easy to see that a HIST exists if and only if its order is not three.

We now consider the case of diameter two. 
For graphs with diameter two and at least 10 vertices, Shan and Tsuchiya~\cite{DBLP:journals/jgt/ShanT23} presented a complete characterization of graphs having a HIST, which is explained below.  

Let $(p_1,\ldots,p_k)$ be a vector consisting of $k$ positive integers, and for $(p_1,\ldots,p_k)$, we define a graph $A(p_1,\ldots,p_k)$ as follows:
\begin{align*}
 V(A(p_1,\ldots,p_k)) & =\{x,y_1,\ldots,y_k\} \cup \bigcup_{i=1}^k U_i, \text{where } U_i = \{u^{(i)}_j \mid j=1,\ldots,p_i\}, \\ 
 E(A(p_1,\ldots,p_k)) & = \bigcup_{i=1}^k \left\{\{x,u\}, \{y_i,u\} \ \middle| \ u\in U_i \right\} \cup \left\{\{y_a,y_b\} \in \binom{\{y_1,\ldots,y_k\}}{2} \right\}. 
\end{align*}
Namely, in $A(p_1,\ldots,p_k)$, $\{y_1,\ldots,y_k\}$ forms a clique with order $k$, and each $G[U_i \cup \{x,y_i\}]$ forms a complete bipartite graph, and each vertices in $U_i$ for each $i$ are twins. Let $\mathcal{A}=\bigcup_{k\in \mathbb{Z}^+}\{A(p_1,\ldots,p_k)\mid p_i \in \mathbb{Z}^+ \text{for each }i\}$. Furthermore, we define $B_n$ by $V(B_n)=V(A(2,n-5))$ and $E(B_n)=E(A(2,n-5))\cup \{\{u^{(1)}_1,u^{(1)}_2\}\}$. The following lemma is known.  

\begin{lemma}[{\cite[Theorem 1]{DBLP:journals/jgt/ShanT23}}]
Let $G$ be a graph of order $n \geq 10$ and diameter $2$. Then $G$ contains a HIST if and only if $G \notin \mathcal{A} \cup \{{B}_n\}$.
\end{lemma}

This lemma implies that for a given $G$, we can determine the existence of a HIST by checking the isomorphism of a graph in $\mathcal{A} \cup \{{B}_n\}$, which is easy indeed. 
For example, we consider checking if a given graph $G=(V,E)$ is isomorphic to a graph forming $A(p_1,\ldots,p_k)$ for a specific $k\ge 3$. Then, we first let $U$ be the set of vertices with degree 2. If $U$ is not an independent set of $G$, then $G$ is not isomorphic to any graph supposed. 
Check if $|\{x\mid N(x)=U\}|\neq 1$. If yes, $G$ is not isomorphic to any graph supposed, and otherwise, let $x$ be the unique vertex such that $N(x)=U$. We then let $Y=V\setminus (U\cup \{x\})$, and if $Y$ forms a clique with order $k$, 
$G$ is isomorphic to a graph forming $A(p_1,\ldots,p_k)$, 
where $p_i = d(y_i) - (k-1)$ for $y_i \in Y$. Otherwise, again $G$ is not isomorphic to any graph supposed. For the other cases, i.e., $A(p_1,\ldots,p_k)$ with $k=1, 2$ and $B_{n}$, we can easily check the isomorphism by modified procedures. 
These procedures are done in polynomial time. 

When we cannot use the lemma, that is, in the case where the number of vertices is at most 9, the existence of a HIST can be determined in constant time. Thus, we have the following theorem: 

\begin{theorem}\label{thm:diameter2}
For a graph $G$ with diameter at most 2, the existence of a HIST can be determined in polynomial time. 
\end{theorem}

\section{Chordal Graphs with Diameter 3}
In this section, we provide a characterization of chordal graphs with diameter 3 that have a HIST. To this end, we first characterize block split graphs that admit a HIST and then extend this characterization to include split graphs and chordal graphs with diameter 3 that admit a HIST. Note that split graphs are one of the most well-studied subclasses of chordal graphs, and they have a diameter of at most 3. We start with block-split graphs. 
\subsection{Block-split graph having a HIST}
Let \( G = (C, I, E) \) be a block-split graph. A vertex \( u \in C \) is called  \emph{good} if its degree satisfies \( d(u) \neq |C| \), \emph{bad} otherwise. A good vertex has either 0 or at least 2 neighbors in \( I \), each of which is a pendant vertex (i.e., has degree 1). A bad vertex has exactly one pendant neighbor in $I$. 
Note that every vertex in $C$ is categorized as good or bad, and the terms ``good'' and ``bad'' are used only for vertices in $C$. 
The following lemma characterizes block-split graphs containing a HIST.

\begin{theorem}\label{thm:blocksplit}
    Let \( G = (C, I, E) \) be a block-split graph that is not a tree and not isomorphic to \( K_3 \). Then \( G \) admits a HIST if and only if it contains at least two good vertices.
\end{theorem}

\begin{proof}
Since \( G \) is not a tree, it follows that \( |C| \geq 3 \).

\medskip 

\noindent
(\(\Rightarrow\)) We prove this by contradiction. Assume that \( G \) has a HIST and at most one good vertex, and \( |C| \geq 3 \); thus, at least two bad vertices exist.  
Note that each bad vertex is adjacent to exactly one pendant vertex, which implies that it should appear as an internal vertex in a HIST.  
Consider a path $(s, t)$ with the longest length in the HIST, where $s$ and $t$ are leaves there. 
This implies that neither $s$ nor $t$ is bad. We consider two cases: (i) neither $s$ nor $t$ is good (i.e., both $s$ and $t$ are pendant vertices in $I$), and (ii) $s$ is a pendant vertex in $I$ and $t$ is good. We first consider (i). Consider the next vertex \( s' \) and $t'$ of the leaves \( s \) and $t$ on the $(s,t)$-path, respectively. Since $s'$ and $t'$ are in $C$, the cases are further divided into (i-1) $s'$ is bad (resp., good) and $t'$ is good (resp., bad), (i-2) $s'$ and $t'$ are bad, and (i-3) $s'(=t')$ is good. (i-1) If $s'$ is bad, $s'$ connects at least two vertices other than $s$ in the HIST, 
which means that $s'$ is adjacent to a vertex $s''$ not in the $s$-$t$ path in the HIST. 
Since the $s$-$t$ path is the longest in the HIST, $s''$ should be a good vertex without a pendant vertex, but it contradicts that $t'$ is good. (i-2) We consider $s''$ and $t''$ adjacent to vertex $s'$ and $t'$ not on the $s$-$t$ path in the HIST as the argument of (i-1). Then, $s''$ and $t''$ are different but must be good. This contradicts the uniqueness of a good vertex. (i-3) Since $(s,t)$-path is the longest, the HIST forms a star, which contradicts the existence of a bad vertex. We next consider (ii). By a similar argument to (i-1), $s'$ has a leaf $s''$, the unique good vertex, on the HIST, which implies $s''=t$. That is, the HIST forms a star, which again leads to the contradiction as (i-3). 

\medskip
\noindent
$( \Leftarrow )$ Suppose $G$ contains at least two good vertices. Then, we can construct a HIST as follows.

We first consider the case where no bad vertex exists; all the vertices in $C$ are good. We further divide the cases according to $|C|$. If $|C|=3$, $C$ contains at least one vertex having at least two pendant vertices, because $G$ is not $K_3$. 
We then fix such a vertex as the root, attach the other vertices in $C$ as the root's children, and attach every pendant vertex to its unique neighbor. The resulting spanning tree is a HIST indeed, because the root has four children and the other vertices in $C$ have no child or at least two children. 
If $|C|\ge 4$, each vertex in $C$ has degree at least 3. Then, we arbitrarily select a vertex in $C$ as the root, attach the other vertices in $C$ as the root's children, and attach every pendant vertex to its unique neighbor. The resulting spanning tree is a HIST again. Thus, if no bad vertex exists, $G$ has at least three good vertices and a HIST. 

We then consider the case where at least one bad vertex exists. 
Since \( G \) contains at least two good vertices by assumption, let \( s \) and \( t \) be two such vertices. We construct a path starting at \( s \), ending at \( t \), and passing through all bad vertices. Note that the path consists of at least three vertices. Let \( r \) be the neighbor of \( s \) on this path, and choose it as the root. Then, connect all the remaining good vertices except \( s \) and \( t \) directly to \( r \). Every pendant vertex is attached to its unique neighbor. Then, this tree is a HIST. Indeed, each bad vertex other than $r$ appears as an internal vertex of the \( s \)-\( t \) path, and thus its degree becomes exactly 3. The degree of $r$ is at least 3. The vertices \( s \) and \( t \) connect to their parent and may have 0 or at least 2 pendant neighbors, resulting in degree 1 or 3 or more. The remaining good vertices have degree 1 or at least 3, depending on the number of pendant neighbors. Thus, all vertices in the tree have degree 1 or at least 3.
\qed
\end{proof}

Since the existence or non-existence of a HIST is trivial for trees and \( K_3 \), we obtain the following corollary. 
\begin{corollary}
For a block split graph $G = (C, I, E)$, we can determine whether $G$ has a HIST in linear time.
\end{corollary}

In a graph $G=(V,E)$, a subgraph $F=(V,E')$ is called a \emph{homeomorphically irreducible spanning forest} (HISF) if $F$ contains neither a cycle nor a vertex with degree two. 
Then the necessary condition of $G$ having a HIST in \Cref{thm:blocksplit} is easily extended to a HISF with $k$ connected components.  

\begin{theorem}\label{thm:blocksplit2}
    Let \( G = (C, I, E) \) be a block-split graph that is not a tree and not isomorphic to \( K_3 \). If \( G \) admits a HISF with $k$ connected components, then it contains at least $2k$ good vertices.
\end{theorem}

\subsection{Split graph having a HIST}
Let \( G = (C, I, E) \) be a split graph. By fully utilizing \Cref{thm:blocksplit,thm:blocksplit2}, we can prove the following theorem, which characterizes split graphs without a HIST (and thus, those with one).
\begin{theorem}\label{thm:split} 
Let \( G = (C, I, E) \) be a split graph that is not a block-split graph.
Then, $G$ does not admit a HIST if and only if one of the following holds:
\begin{enumerate}
 \item For all $u \in C$, $|N(u)\cap I|=1$ and $|C|-|I|= 1$. 
 \item Let $U = \{u \in C \mid |N(u) \cap I| = 2\}$. All of the following hold:
 \begin{enumerate}
 \item For all $u \in C$, $|N(u) \cap I| \in \{1,2\}$ and $|U| \geq 2$,
 \item For each $u \in U$, the neighborhoods $N(u)\cap I$ are distinct,
 \item For each $u \in C\setminus U$, the vertex in $N(u)$ is a pendant vertex,
 \item $\left|\bigcap_{u \in U} (N(u)\cap I)\right| = 1$.
 \end{enumerate}
\end{enumerate}
\end{theorem}

\begin{proof}
($\Rightarrow$)  
We prove that if $G$ satisfies neither condition \textbf{1} nor condition \textbf{2}, then $G$ admits a HIST.  
We first focus on the possible values of $|N(u)\cap I|$.  
Conditions \textbf{1} and \textbf{2} assume that $|N(u)\cap I| \in \{1,2\}$ for all $u \in C$.  
Hence, we consider the following two cases: (A) There exists a vertex $u$ with $|N(u)\cap I| \notin \{1,2\}$, (B) For all $u \in C$, $|N(u)\cap I| \in \{1,2\}$.

\noindent\textbf{Case (A):}  
First, consider the situation where there exists a vertex $u$ with $|N(u)\cap I| = 0$.  
Two subcases arise:  
(a) There exists a vertex $v \in C$ (distinct from $u$) with $|N(v)\cap I| \neq 1$. In this case, we can delete edges 
incident to neighbours of $v$ in $I$ (if exist) to ensure that such vertices become pendant vertices (this deletion does not affect the degree of $u$).  
By further deleting edges appropriately to make the degree of vertices in $I$ equal to 1, we obtain a block split graph $G'$ with good vertices $u$ and $v$, implying that $G$ admits a HIST. 
(b) If all vertices $v \in C \setminus \{u\}$ satisfy $|N(v)\cap I| = 1$, since $G$ is not a block split graph, there must exist a vertex $w \in I$ with degree at least 2.  
This vertex $w$ must be the unique neighbor in $I$ of its adjacent vertex in $C$.  
By retaining only one edge incident to $w$, another vertex in $C$ will have degree 0 with respect to $I$.  
Further deleting edges to ensure that all remaining vertices in $I$ have degree 1 results in a block split graph $G'$ with at least two good vertices, which contains a HIST by \Cref{thm:blocksplit}.

Next, consider the case where there exists a vertex $u$ with $|N(u)\cap I| \geq 3$.  
Select $u^*$ to maximize $|N(u^*)\cap I|$.  
Two main subcases arise: (c) There exists another vertex $v \in C$ with $|N(v)\cap I| \geq 2$, (d) Otherwise. Case (c) further divides into:
\begin{itemize}
    \item (c-1) $N(v)\cap N(u^*) \cap I = \emptyset$
    \item (c-2) $N(v)\subseteq N(u^*)$
    \item (c-3) $|N(v)\cap N(u^*) \cap I| = 1$
    \item (c-4) $N(v)\not\subseteq N(u^*)$ and $|N(v)\cap N(u^*) \cap I| \geq 2$
\end{itemize}
In all these cases, by making the neighbors in $I$ of $u^*$ and $v$ pendant vertices connected only to $u^*$ or $v$, and by appropriately connecting the remaining vertices in $I$ to $C$ as pendants, we obtain a block split graph with good vertices $u^*$ and $v$, implying the existence of a HIST. We examine each case in order. 

In case (c-1), we construct pendant vertices by retaining all edges connecting \( N(u^*) \cap I \) to \( u^* \) and all edges connecting \( N(v) \cap I \) to \( v \). In case (c-2), we construct pendant vertices by retaining all edges connecting \( N(u^*) \cap I \) to \( u^* \), treating each vertex in \( N(u^*) \cap I \) as a pendant vertex. In case (c-3), we construct pendant vertices by retaining all edges connecting \( N(v) \cap I \) to \( v \) and all edges connecting \( N(u^*) \cap I \setminus N(v) \) to \( u^* \), where \( N(u^*) \cap I \setminus N(v) \) contains at least two vertices.
In case (c-4), we have \( |N(u^*) \cap I| \ge |N(v) \cap I| \ge 3 \), \( |(N(u^*) \setminus N(v)) \cap I| \ge 1 \), and \( |(N(v) \setminus N(u^*)) \cap I| \ge 1 \). Hence, we construct pendant vertices by connecting the vertices in \( (N(u^*) \setminus N(v)) \cap I \) to \( u^* \), the vertices in \( (N(v) \setminus N(u^*)) \cap I \) to \( v \), and at least one vertex in \( N(u^*) \cap N(v) \cap I \) to each of \( u^* \) and \( v \). The remaining vertices in \( (N(u^*) \cup N(v)) \cap I \) can be connected to either \( u^* \) or \( v \) arbitrarily. With this construction, both \( u^* \) and \( v \) become good vertices.

For subcase (d), the argument proceeds similarly to (b) and also yields a HIST.

\smallskip

\noindent
\textbf{Case (B):}  
Let $U = \{ u \in C \mid |N(u)\cap I| = 2\}$.  
First, if $|U|=0$, i.e., all $u \in C$ satisfy $|N(u)\cap I|=1$, then as long as $|C|-|I| \neq 1$ (i.e., condition \textbf{1} does not hold), we can show that $G$ admits a HIST.  
The condition that all $u \in C$ satisfy $|N(u)\cap I|=1$ implies $|C| = \sum_{v \in I} d(v)$. 
By the connectivity of the split graph, we have \( \bigcup_{u \in C} N(u) \cap I = I \).  
Since \( |N(u) \cap I| = 1 \) holds for every \( u \in C \), it follows that
\[
|I|=\left| \bigcup_{u \in C} N(u) \cap I \right| \leq \sum_{u \in C} |N(u) \cap I| =\sum_{u \in C} 1 = |C|.
\]
In other words, if \( |C| - |I| \neq 1 \), then either \( |C| - |I| = 0 \) or \( |C| - |I| \geq 2 \) holds.

When $0=|C| - |I| = \sum_{v \in I} d(v) - |I|$, we have \( d(v) = 1 \) for every \( v \in I \).  In this case, the graph becomes a block split graph, which is not allowed.  
Thus, we now focus on the case where \( |C| - |I| \geq 2 \).

Since the equation  
\[
|C| = \sum_{v \in I} d(v) = \sum_{v \in I} 1 + \sum_{v \in I} (d(v) - 1) = |I|  + \sum_{v \in I} (d(v) - 1),  
\]
holds, this situation implies that the set \( I \) must contain either at least two vertices of degree at least two, or at least one vertex of degree at least three.  
In either case, by deleting all but one edge incident to each of these vertices on the \( C \) side, we can construct a block split graph containing two good vertices, which admits a HIST.  

When $|U| \geq 1$, the remaining cases (i) $|U| = 1$, corresponding condition \textbf{2.a}, (ii) $\exists u, u' \in U$ with $N(u)\cap I = N(u')\cap I$, corresponding \textbf{2.b}, (iii) $\exists u \in U, u' \in C\setminus U$ with $N(u')\subseteq N(u)$, corresponding condition \textbf{2.c}, and (iv) $\left|\bigcap_{u \in U} (N(u)\cap I)\right| = 0$, corresponding condition \textbf{2.d} (excluding the case already covered by (ii)) can all be handled by appropriate edge deletions to construct a block split graph with at least two good vertices, ensuring the existence of a HIST.

First, we consider case (i). Let $u \in C$ be the vertex such that $|N(u) \cap I| = 2$. Since $G$ is not a block split graph, there exists a vertex $v \in I$ with \( d(v) \geq 2 \). If \( u \in N(v) \), we retain the edge connecting \( v \) to \( u \) as a pendant edge and remove the other edges incident to \( v \). If \( u \notin N(v) \), we arbitrarily select one vertex in \( N(v) \setminus \{u\} \) and remove all edges incident to \( v \) except for the one to the selected vertex. In this way, we obtain a block split graph in which the selected vertex from \( N(v) \setminus \{u\} \) and \( u \) are good vertices. Since this graph admits a HIST, so does \( G \).

In case (ii), we retain only the edges connecting \( N(u) \cap I (= N(u') \cap I) \) to \( u \) (removing the edges connecting to \( u' \)) and adjust the degrees of the remaining vertices in \( I \) to be one. The resulting graph is a block split graph containing the good vertices \( u \) and \( u' \), and therefore admits a HIST.
Also, in case (iii),  we remove the edges connecting $N(u') \cap I$ and apply the same argument above, and then the resulting graph admits a HIST. 

In case (iv), it suffices to consider the situation where \( \left| \bigcap_{u \in U} (N(u) \cap I) \right| = 0 \) since the case where its size is $2$ has already been addressed in case (ii). We proceed by considering the possible sizes of \( U \).
For the case of \( |U| = 2 \), let \( U = \{u, u'\} \). Since \( N(u) \cap N(u') \cap I = \emptyset \), we connect the vertices in \( N(u) \cap I \) to \( u \) and those in \( N(u') \cap I \) to \( u' \) as pendant vertices, and connect the remaining vertices in \( I \) arbitrarily as pendant vertices. The resulting graph is a block split graph with good vertices \( u \) and \( u' \), which implies that it admits a HIST.

Next, we consider the case \( |U| \geq 3 \). If there exist \( u, u' \in U \) such that \( N(u) \cap N(u') \cap I = \emptyset \), then by the same argument as above, \( G \) admits a HIST. Otherwise, we assume that \( N(u) \cap N(u') \cap I \neq \emptyset \) holds for all \( u, u' \in U \); we refer to this property as the \emph{commonality} condition.
Let us consider a vertex \( u^* \in U \) with \( N(u^*) = \{v_1, v_2\} \). By the commonality condition, there must exist a vertex \( u_1 \in U \) such that \( v_1 \in N(u_1) \) and a vertex \( u_2 \in U \) such that \( v_2 \in N(u_2) \). (If such \( u_1 \) does not exist, then by the commonality condition, every vertex \( u \in U \) must satisfy \( v_2 \in N(u) \), which contradicts our assumption. The existence of \( u_2 \) can be argued similarly.)
By the commonality condition between \( u_1 \) and \( u_2 \), there must exist \( v_3 \in N(u_1) \cap N(u_2) \cap I \). In this configuration, we have \( N(u_1) = \{v_1, v_3\} \) and \( N(u_2) = \{v_2, v_3\} \). Due to the distinctness of the vertices, this case can only occur when \( |U| = 3 \).
In this case, we construct a block split graph by connecting \( v_1 \) and \( v_2 \) to \( u^* \) as pendant vertices, connecting \( v_3 \) to \( u_2 \) as a pendant vertex, and adjusting the remaining vertices in \( I \) as pendant vertices appropriately. This graph contains \( u^* \) and \( u_1 \) as good vertices, and therefore admits a HIST.

\smallskip 

\noindent($\Leftarrow$)  
We first show that if $G = (C, I, E)$ satisfies condition \textbf{1}, then it does not admit a HIST.  
Since each vertex in $C$ is adjacent to exactly one vertex in $I$, the number of edges between $C$ and $I$ is $\sum_{u \in C} |N(u) \cap I| = |C|$. On the other hand, since the graph is connected, every vertex $v \in I$ must satisfy $d(v) \geq 1$, and thus we have:
\[
|I|+1=|C| = \sum_{v \in I} d(v) = \sum_{v \in I} 1 + \sum_{v \in I} (d(v) - 1) \geq |I| + \sum_{v \in I: d(v) > 1} 1.
\]
This implies that exactly one vertex $v^* \in I$ satisfies $d(v^*) = 2$ and all other vertices in $I$ have degree 1.  
Suppose, for contradiction, that $G$ admits a HIST. Since $d(v^*) = 2$, the degree of $v^*$ must be 1 in a HIST, meaning that only one of its two incident edges is included in the HIST.  
This implies that the HIST of $G$ is also a HIST of the graph $G'$ obtained by removing one of the edges incident to $v^*$. However, such a graph $G'$ would then become a block split graph with exactly one good vertex, contradicting \Cref{thm:blocksplit}.  

Next, we show that \( G = (C, I, E) \) does not admit a HIST when it satisfies Condition \textbf{2}.

Let $v^*$ denote the unique element of $\bigcap_{u \in U} N(u) \cap I$. Then,  
$N(u)\cap I \setminus \{u^*\}$ for $u \in U$ are singletons and distinct by condition \textbf{2.b.}, and each of them is not connected to any vertex in $C\setminus U$ by condition \textbf{2.c.}; all of them are pendant vertices and $u^*$ is the unique non-pendant vertex in $G$.  
Suppose, for contradiction, that $G$ admits a HIST $T$. Then $v^*$ must either (i) be a leaf in \( T \), or (ii) be an internal vertex of degree at least 3.

\begin{itemize}
    \item[(i)] Assume that \( v^* \) is a leaf in \( T \), and let \( u' \) denote the vertex adjacent to \( v^* \) in \( T \). Consider the graph obtained by removing all edges \( \{\{u, v^*\} \mid u \in C \setminus \{u'\}, |N(u) \cap I| = 2\} \) from \( G \). Since the veritces in $I$ other than $u^*$ are pendants in $G$, the resulting graph is a block split graph. Moreover, \( T \) remains a HIST in this block split graph. However, this block split graph has no good vertices other than \( u' \), and thus it cannot admit a HIST, leading to a contradiction.

    \item[(ii)] Assume that \( v^* \) is an internal vertex of degree \( l \geq 3 \) in \( T \). Let \( N_T(v^*) = \{u_1, \ldots, u_l\} \subseteq U \). We construct a new graph \( G' = (V \setminus \{v^*\} \cup \{v^*_1, \ldots, v^*_l\}, E \cup \bigcup_{i = 1}^l \{\{u_i, v^*_i\}\}) \) by replacing \( v^* \) with \( l \) copies \( v^*_1, \ldots, v^*_l \), and connecting each \( u_i \) to the corresponding \( v^*_i \). In this graph \( G' \), all vertices \( u \in C \setminus \{u_1, \ldots, u_l\} \) satisfy \( |N(u) \cap I| = 1 \), making \( G' \) a block split graph that contains exactly $l$ good vertices. On the other hand, since each edge \( \{u_i, v^*\} \) in \( T \) has been replaced by \( \{u_i, v^*_i\} \), \( G' \) contains an HISF with \( l \) connected components. This contradicts \Cref{thm:blocksplit2}.
\end{itemize}
Therefore, \( G \) does not admit a HIST under Condition \textbf{2}.
\qed
\end{proof}

Since the conditions of \Cref{thm:split} can be checked in linear time, we have the following corollary. 
\begin{corollary}
\label{cor:split}
Given a split graph \( G = (C, I, E) \), we can determine whether $G$ has a HIST in linear time.
\end{corollary}

\subsection{Chordal graphs with diameter 3 having a HIST}
We are now ready to give a characterization of chordal graphs with diameter 3 
that admit a HIST. Before going to the proof, we first see a property of chordal graphs with diameter 3. 
\begin{lemma}[{\cite[Theorem 2.1]{DBLP:journals/dm/KratschDL94}}]\label{lem:dclique}
A chordal graph $G$ has a dominating clique if and only if it has a diameter at most 3.
\end{lemma}
Let $G$ be a chordal graph with diameter 3 that is not a split graph. By \Cref{lem:dclique}, $G$ has a (maximal) dominating clique, say $C$. Let $G=(C,\bar{C},E)$, where $\bar{C}=V\setminus C$. Here, $G[\bar{C}]$ contains at least one edge, because otherwise it becomes a split graph. 

\begin{theorem}\label{thm:d3}
Let \( G = (C, \bar{C}, E) \) be a chordal graph with diameter 3 that is not a split graph, and let $U = \{u \in C \mid |N(u) \cap \bar{C}| = 2\}$.
Then, $G$ does not admit a HIST if and only if one of the following holds:
\begin{enumerate}
 \item $\exists u^* \in C: |N(u^*)\cap \bar{C}|\ge 3$, and $\forall u\in C\setminus \{u^*\}: |N(u)\cap \bar{C}|=1$ and the neighbor of $u$ in $\bar{C}$ is a pendant vertex.  
 \item $\forall u \in C: |N(u)\cap \bar{C}|\in \{1,2\}$, every vertex in $\bigcup_{u\in C\setminus U} N(u)$ is a pendant vertex, and
    \begin{enumerate}
        \item $|U|=1$, or 
        \item $|U|=2$ and $\left|\bigcap_{u \in U} (N(u)\cap \bar{C})\right| = 1$, or 
        \item $|U|\ge 3$, for each $u\in U$, the neighborhoods $N(u)\cap \bar{C}$ are distinct, $\left|\bigcap_{u \in U} (N(u)\cap \bar{C})\right| = 1$, and $G[\bar{C}]$ contains 1 edge. 
 \end{enumerate}
\end{enumerate}
\end{theorem}
\begin{proof}
($\Rightarrow$) If neither condition \textbf{1} nor condition \textbf{2} holds, it suffices to show that the graph admits a HIST. A case that clearly violates both conditions is when there exists a vertex $u^* \in \bar{C}$ such that $|N(u^*) \cap \bar{C}| = 0$. First, if the subgraph of $G$ obtained by deleting edges within $G[\bar{C}]$ is not a block-split graph but a split graph, then it evidently does not satisfy the conditions of \Cref{thm:split}, and thus it admits a HIST. Thus, we consider the case where the resulting graph is a block-split graph. In such a case, suppose that there is no vertex $u$ such that $N(u)$ contains both endpoints of a deleted edge $\{v, v'\}$. Then, by selecting $v \in N(u)$ and $v' \in N(u')$, the sequence $(u, v, v', u, u)$ forms a chordless cycle of length four, which contradicts the chordality of the graph. Hence, since $G[\bar{C}]$ contains at least one edge, there must exist a vertex $u$ such that $N(u)$ contains some $\{v, v'\}$. This implies that $|N(u) \cap \bar{C}| \geq 2$, and therefore, the block-split graph obtained by deleting edges within $G[\bar{C}]$ contains good vertices $u^*$ and $u$, and thus admits a HIST.

In the following, we thus consider the case where all vertices in $C$ satisfy $|N(u) \cap \bar{C}| \geq 1$.
If there exists a vertex $u_1 \in C$ such that $|N(u_1) \cap \bar{C}| \geq 3$, the cases where condition \textbf{1} is not satisfied are when there exists another vertex $u_2 \in C \setminus \{u_1\}$ such that $|N(u_2) \cap \bar{C}| \geq 2$, or when $\forall u\in C\setminus \{u^*\}: |N(u)\cap \bar{C}|=1$ but $\exists u' \in C\setminus \{u^*\}:$ the neighbor $v$ of $u'$ in $\bar{C}$ has $d(v)\ge 2$. In the former case, the block-split graph or split graph obtained by deleting edges within $G[\bar{C}]$ admits a HIST by \Cref{thm:split,thm:blocksplit}. In the latter case, note that the edges in $G[\bar{C}]$ are only in $N(u^*)\cap \bar{C}$ due to the chordality. Thus, by deleting edge $\{u',v\}$ and edges in $G[\bar{C}]$, we again obtain a block-split graph or split graph that admits a HIST by \Cref{thm:split,thm:blocksplit}.

If there is no vertex $u_1 \in C$ such that $|N(u_1) \cap \bar{C}| \geq 3$, then all vertices satisfy $|N(u) \cap I| \in \{1, 2\}$, which corresponds to condition \textbf{2}.  
We consider this case based on the value of $|U|$, where $U = \{u \in C \mid |N(u) \cap I| = 2\}$. Since condition \textbf{2} also forces that every vertex in $\bigcup_{u\in C\setminus U} N(u)$ is a pendant vertex, we consider the case where there exists a vertex $v^*$ in $\bigcup_{u\in C\setminus U} N(u)$ that is not a pendant vertex. Note that $G[\bar{C}]$ containing an edge implies $|U|>0$, because otherwise it violates chordality. 
Thus, there is a vertex $u'$ with $|N(u') \cap \bar{C}|=2$. Then, by removing edges incident to $v^*$ except one and edges in $G[\bar{C}]$, we obtain a (block) split graph having a HIST by \Cref{thm:blocksplit,thm:split}. 
Thus, in the following, we assume that every vertex in $\bigcup_{u\in C\setminus U} (N(u)\cap I)$ is a pendant vertex.  

If $|U| \geq 2$, condition \textbf{2} is not satisfied if either $\bigcap_{u \in U} (N(u) \cap \bar{C}) = \emptyset$ or there exist $u, u' \in U$ such that $N(u) \cap \bar{C}= N(u') \cap \bar{C}$.  
In this case, the split graph or block-split graph obtained by deleting edges within $G[\bar{C}]$ admits a HIST by \Cref{thm:blocksplit,thm:split}.  

Finally, we consider the case where $|U| \geq 3$, the neighborhoods $N(u) \cap \bar{C}$ for $u \in U$ are pairwise distinct, and $\bigcap_{u \in U} (N(u) \cap \bar{C}) = 1$.  
Condition \textbf{2} is not satisfied in this case if $G[\bar{C}]$ contains at least two edges.  
Without loss of generality, assume that $\{u_1, u_2, u_3\} \subseteq U$ satisfy $N(u_1) = \{v^*, v_1\}$, $N(u_2) = \{v^*, v_2\}$, and $N(u_3) = \{v^*, v_3\}$, and that $\{v^*, v_1\}, \{v^*, v_2\} \in E$.  
In this case, by deleting the edges $\{u_1, v_1\}$, $\{u_1, v^*\}$, $\{u_2, v_2\}$, and $\{u_2, v^*\}$, while retaining the edge $\{u_3, v^*\}$, and deleting all remaining edges in $G[\bar{C}]$, the resulting graph admits a HIST.  
Indeed, the graph obtained by removing vertices $u_1$ and $u_2$, the edges $\{u_1, v^*\}$ and $\{u_2, v^*\}$, and all edges within $G[\bar{C}]$, and then adding the edges $\{v_1, v^*\}$ and $\{v_2, v^*\}$ to a HIST of this reduced graph, constitutes a HIST of the original graph $G$.  

From the above, we have shown that if neither condition 1 nor condition 2 holds, then $G$ admits a HIST.

\medskip 

\noindent ($\Leftarrow$) We first show that if $G = (C, \bar{C}, E)$ satisfies condition 1, then it does not admit a HIST. Since $G$ is not a split graph, $G[\bar{C}]$ must contain at least one edge. By chordality, such an edge must exist between vertices in $N(u^*) \cap I$.  Now, suppose there exists a HIST in which all vertices in $N(u^*) \cap I$ are leaves. In this case, the HIST is also a HIST of the block-split graph obtained by deleting the edges within $N(u^*) \cap I$ from $G$. However, this contradicts \Cref{thm:blocksplit}.  
Therefore, if $G$ admits a HIST, the HIST must contain at least one internal vertex of degree at least three among the vertices in $N(u^*) \cap I$. However, by an argument similar to case ($\Leftarrow$) (ii) in the proof of \Cref{thm:split}, it can be seen that even in this case, the graph does not admit a HIST.

Next, we will show that if $G = (C, \bar{C}, E)$ satisfies condition 2, then it does not admit a HIST. Note that the edge in $G[\bar{C}]$ must be in $\bigcup_{u\in U} N(u)\cap I$ by the same argument above. We start from condition 2(a). This case is almost trivial. 
Since the vertices in $\bigcup_{u\in U} N(u)\cap I$ have degree 2, they must be leaves in any HIST; it is essentially a block-split graph with no HIST. 

We then go to graphs satisfying condition 2(b). Let $v^*$ be the unique vertex of $\bigcap_{u\in U} (N(u)\cap I)$, which is the only vertex that can be an internal vertex with degree at least $3$ in a HIST $T$. Let $U=\{u_1,u_2\}$, $N(u_1)=\{v^*,v_1\}$, and $N(u_2)=\{v^*,v_2\}$. There are three essential cases (i) $N_T(v^*)=\{v_1,v_2,u_1\}$,  
(ii) $N_T(v^*)=\{v_1,u_1,u_2\}$,  and (iii) $N_T(v^*)=\{v_1,v_2,u_1,u_2\}$. 
In case (i), if $G$ has a HIST satisfying this condition, graph $G'=(V\setminus \{v_1,v_2\}, E\setminus \{\{u_2,v^*\}\})$ does so, but it is a block-split graph with one good vertex, leading to a contradiction. We consider case (ii). If $u_1$ is a leaf in a HIST, the case is reduced to $G\setminus \{u_1\}$, concluding that it does not admit a HIST. Otherwise, by an argument similar to case ($\Leftarrow$) (ii) in the proof of \Cref{thm:split}, the case is reduced to whether a graph $G'$ transformed from $G$ has a HISF with $3$ components, concluding that $G$ has no HIST satisfying the condition.  
Case (iii) is also confirmed like (ii). 

Finally, we consider graphs satisfying condition 2(c). 
Suppose that $v^*$ is the unique vertex of $\bigcap_{u\in U} (N(u)\cap I)$, $U=\{u_1,\ldots,u_{|U|}\}$, and $\{v_1,v^*\}$ is the unique edge in $G[\bar{C}]$. In this case also, $v^*$ is the only vertex that can be an internal vertex with degree $l (\ge 3)$ in a HIST $T$. There are two essential cases: (a) $N_T(v^*)=\{v_1, u_1,\ldots,u_{l-1} \}$, and  (b) $N_T(v^*)=\{v_1, u_2,\ldots,u_{l}\}$. Case (a) can be considered as condition 2 (b) (ii); the case divides into that in a HIST $u_1$ is a leaf, or not, and we can conclude that $G$ has no HIST satisfying the condition. In case (b), $u_1$ can become a good vertex in a graph transformed from $G$, but the number of connected components is more; we can conclude that $G$ has no HIST satisfying the condition again. Overall,  graphs satisfying condition 2(c) has no HIST. \qed
\end{proof}

Note that finding a dominating clique of a chordal graph can be done in polynomial time, because we can enumerate all maximal cliques of a chordal graph in linear time~\cite{DBLP:journals/siamcomp/Gavril72}. 
Given a maximal dominating clique, the conditions of \Cref{thm:d3} can be easily checked, which implies the following corollary. 

\begin{corollary}
Given a chordal graph \( G\) with diameter at most $3$, we can determine whether $G$ has a HIST in polynomial time.
\end{corollary}

\section{Hardness results}
In this section, we give an NP-hardness proof, which yields several hardness results. 

What we show in this section is the NP-hardness of deciding whether a given strongly chordal graph with diameter at most 4 admits a HIST. A \emph{chordal bipartite graph} is a bipartite graph in which every induced cycle of length at least 6 has a chord\cite{DBLP:journals/dm/Muller96a}.
To prove this, we use the fact that it is NP-hard to decide  whether a given chordal bipartite graph $G=(U,V,E)$ with two pendant vertices $s\in U,t\in V$ and $|U|=|V|$ admits a Hamiltonian path~\cite{DBLP:journals/dm/Muller96a}. 
We first construct $G'$ from $G$ by adding two new vertices $s'$ and $s''$ as follows: $G'=(U\cup \{s'\}, V\cup \{s''\}, E\cup \{\{s,s''\}\}\cup \bigcup_{v \in V\cup \{s''\}\setminus \{t\}} \{s',v\})$. This modification does not yield a new induced cycle with length at least $6$. In fact, since a cycle in $G'$ but not in $G$ must contain $s'$, we focus on a cycle containing $s'$; a cycle containing $s'$ with length at least 6 contains at least two vertices in $V$, but $s'$ is adjacent to them, which spans a chord. Therefore, $G'$ is still a chordal bipartite graph. We also see that $G'$ has an $s'$-$t$ Hamiltonian path if and only if $G$ has an $s$-$t$ Hamiltonian path. Here, the if-direction is obvious, so we consider the only-if direction. 
Suppose $G'$ has an $s'$-$t$ Hamiltonian path. If it goes as $(s',s'',s,\ldots,t',t)$, it contains an $s$-$t$ Hamiltonian path in $G$. Otherwise, the path starts from $s$, and then goes to a vertex in $V$. However, then $s''$ cannot be passed, and such a case never happens. 

We next construct $G''$. Its vertex set is the same as $G'$, that is, $V(G''):=V(G')$. The edge set $E(G''):=E(G')\cup \binom{U\cup \{s'\}}{2}$, that is, $U\cup \{s'\}$ forms a clique; $G''$ is a split graph. Furthermore, we can see that it does not contain a $k$-\emph{sun} for any $k\ge 3$ as an induced subgraph. Here, for $k\ge 3$, a $k$-\emph{sun} is a graph on vertices $X \cup Y$ with $X=\{x_0,\ldots,x_{k-1}\}$ and $Y=\{y_0,\ldots,y_{k-1}\}$ such that $X$ and $Y$ respectively form a clique and an independent set, and for every $i\in \{1,\ldots,k\}$, $y_i$ is adjacent to $x_i$ and $x_{(i+1) \text{mod\ } k}$. If $G''$ contains a $k$-sun for $k\ge 3$, its  clique part $X$ is in $U\cup \{s'\}$ and the independent set part $Y$ is in $V\cup \{s''\}$, and thus $G'$ must contain an induced cycle $(x_0,y_0,x_1,\ldots,x_i,y_i,x_{i+1},\ldots,y_{k-1},x_0)$, whose length is $2k$. This contradicts that $G'$ is chordal bipartite. Since sun-free chordal graphs are strongly chordal~\cite{DBLP:journals/dm/Farber83}, we obtain the following lemma. 

\begin{lemma}\label{lem:scs2hp}
For a given strongly chordal split graph $G$, $s,t\in V(G)$, determining whether $G$ admits an $s$-$t$ Hamiltonian path is NP-complete.
\end{lemma}

\begin{theorem}\label{thm:NP:diameter}
For a strongly chordal graph $G$ of diameter at most 4, determining whether $G$ has a HIST is NP-complete.
\end{theorem}

\begin{proof}
We reduce from the \( s \)-\( t \) Hamiltonian path problem in the strongly chordal split graph \( G' \) constructed in the proof of \Cref{lem:scs2hp}. Let $H$ be the graph obtained from \( G'' \) by adding new pendant vertices as follows: for every vertex in $U\cup V\cup\{s''\}\setminus \{t\}$, attach one pendant vertex, and for $s'$, attach two pendant vertices.  
Note that $H$ is still strongly chordal and its diameter is at most $4$, since $s$ is adjacent to any vertex in $U\cup V\cup\{s''\}\setminus \{t\}$. 

Now we prove that $H$ admits a HIST if and only if \( G' \) has an \( s' \)-\( t \) Hamiltonian path.

\noindent 
(\(\Rightarrow\)) Suppose \( G' \) has an \( s' \)-\( t \) Hamiltonian path $P$. 
Here, the vertices in $P$ are $U\cup V\cup\{s',s''\}$, and by attaching the new pendant vertices in $H$, we obtain a spanning tree of $H$. 
We can see that in this spanning tree, the degrees of $U\cup V\cup\{s',s''\}\setminus \{t\}$ are all 3, and that of $t$ is 1; it is a HIST of $H$. 

\noindent
(\(\Leftarrow\)) Suppose \( H \) admits a HIST \( T \). 
The number of vertices in $H$ is the sum of the vertices in $G'$ and the number of newly added pendant vertices. 
Since the former is $|U\cup V\cup\{s',s''\}|=2|U|+2$, 
and the latter is $|U\cup V\cup\{s',s''\}\setminus \{t\}|+1=2|U|+2$  (because $s'$ has two pendants), the total is $4|U|+4$. 
This implies that the sum of degrees of the vertices in $T$, called the \emph{total degree} of $T$, is $2(4|U|+3)=8|U|+6$ by the handshaking lemma. 
We then consider the tree \( T' \) obtained by removing all pendant vertices other than $t$ from $T$. Since the number of removed pendant vertices is $2|U|+2$, the total degree of $T'$ is $8|U|+6-2(2|U|+2)=4|U|+2$. Here, we count the total degree of $T'$ in another way. Since the vertices in $U\cup V\cup\{s''\}\setminus \{t\}$ are internal vertices in HIST $T$, their degrees are at least 3; their degrees in $T'$ are at least $2$. The sum of degrees of $U\cup V\cup\{s''\}\setminus \{t\}$ in $T'$ is at least $4|U|$. The remaining vertices in $T'$ are $s'$ and $t$, and the sum of their degrees is at least 2; the total degree of $T'$ is at least $4|U|+2$. Since the total degree of $T'$ is exactly $4|U|+2$ as seen above, the degree of every vertex  in $U\cup V\cup\{s''\}\setminus \{t\}$ in $T'$ must be exactly 2, and those of $s'$ and $t$ are exactly 1. This implies that $T'$ is an $s'$-$t$ Hamiltonian path of $G'$. 
\qed 
\end{proof}

It is shown that $s$-$t$ Hamiltonian Path is NP-complete even on chordal bipartite graphs~\cite{tcs/HanakaK25} and planar graphs of maximum
degree 3~\cite{sofsem/MeloFS21}. Moreover, it is W[1]-hard when parameterized by cliquewidth~\cite{tcs/HanakaK25}.
Since our reduction (with a small modification) in \Cref{thm:NP:diameter} is only attaching a pendant vertex to each vertex, the resulting graph keeps chordal bipartiteness, or planarity. 
Also, it increases the maximum degree by one and clique-width by at most one. Thus, we have the following corollary.

\begin{corollary}\label{cor:NP:other}
    Determining whether G has a HIST is NP-complete even on chordal bipartite graphs and planar graphs of maximum degree 4. Moreover, it is W[1]-hard when parameterized by cliquewidth.
\end{corollary}

\section{Exact exponential-time algorithm}
A complete graph with \( n \) vertices has \( n^{n-2} \) spanning trees, and thus any \( n \)-vertex graph has at most \( n^{n-2} \) spanning trees. Since all such spanning trees can be enumerated with constant delay, the HIST existence problem can be solved in time \( O(n^{n-2} \cdot m) \), i.e., within \( 2^{O(n \log n)} \) time. We aim to design faster exact algorithms for this problem.

\begin{algorithm}[H]
\caption{Exact Algorithm for HIST}\label{alg:3}
\KwIn{A graph $G=(V,E)$ with $n$ vertices}
\KwOut{Yes if $G$ has a HIST, No otherwise}

Define $C[S][S_1][S_2] = 1$ if there exists a spanning tree of $G[S]$ where $S_1$ are degree-1 vertices and $S_2$ are degree-2 vertices; otherwise, 0\;

Initialize $C[S][S_1][S_2] \gets 0$ for all $S \subseteq V$, $S_1 \subseteq S$, $S_2 \subseteq S \setminus S_1$\;

\ForEach{$\{u,v\} \in E$}{
    Set $C[\{u,v\}][\{u,v\}][\emptyset] \gets 1$\;
}

\For{$s = 3$ to $n$}{
    \ForEach{subset $S \subseteq V$ with $|S| = s$}{
        \ForEach{$j \in S$, and $S_1 \subseteq S$ with $j \in S_1$, and $S_2 \subseteq S \setminus S_1$}{
            Let $S' = S \setminus \{j\}$ and $S_1' = S_1 \setminus \{j\}$\;
            \ForEach{$k \in N(j) \cap S$}{
                \uIf{$k \in S_2$}{
                    Set $C[S][S_1][S_2] \gets C[S][S_1][S_2] \lor C[S'][S_1' \cup \{k\}][S_2 \setminus \{k\}]$\;
                }
                \ElseIf{$k \in S \setminus (S_1 \cup S_2)$}{
                    Set $C[S][S_1][S_2] \gets C[S][S_1][S_2] \lor C[S'][S_1'][S_2 \cup \{k\}] \lor C[S'][S_1'][S_2]$\;
                }
            }
        }
    }
}

\ForEach{$S_1 \subseteq V$ with $|S_1| \ge \lceil n/2 \rceil$}{
    \If{$C[V][S_1][\emptyset] = 1$}{
        \Return Yes
    }
}
\Return No
\end{algorithm}
\Cref{alg:3} is a dynamic programming algorithm for the HIST existence problem. It defines a function \( C[S][S_1][S_2] \) that indicates whether there exists a spanning tree of the subgraph \( G[S] \), where the set \( S_1 \subseteq S \) consists of vertices of degree 1 and the set \( S_2 \subseteq S \) consists of vertices of degree 2 in the tree.

Such a tree must contain at least one degree-1 vertex. Suppose we choose \( j \in S \) as a leaf. It must connect to some neighbor \( k \in N(j) \cap S \). If \( k \in S_2 \), then the tree with \( j \) and \( k \) must be such that \( k \) is converted to a degree-1 vertex upon removing \( j \). If \( k \) has not yet been assigned a degree, it can be placed into either \( S_2 \) or left unassigned, depending on how the tree grows. These cases correspond to lines 10 and 13 of the algorithm.

Since the algorithm enumerates all subsets \( S \subseteq V \), and their subpartitions \( S_1, S_2 \), it runs in time \( 4^n n^{O(1)}  \).

\begin{theorem}
For an \( n \)-vertex graph, the existence of a HIST can be decided in \( 4^n n^{O(1)}  \) time.
\end{theorem}

This algorithm is rather straightforward, but a modified version can be used in other algorithms as a subroutine.

\section{FPT algorithms by Structural Graph Parameters}
In this section, we investigate the parameterized complexity of the HIST problem with respect to several structural graph parameters. Section~6.1 presents an $O^*(4^k)$-time algorithm parameterized by the modular-width $k$. The algorithm exploits the structure provided by the modular decomposition to normalize potential HISTs, and systematically explores the quotient graph using the exact algorithm developed in Section~5. To verify whether the degree constraints required for a HIST can be satisfied, the algorithm solves a feasibility problem formulated as a system of integer linear constraints. In Section~6.2, we briefly show that the problem is fixed-parameter tractable when parameterized by treewidth by providing an MSO$_2$ formulation. Section~6.3 develops an FPT algorithm parameterized by the cluster vertex deletion number, where we employ kernelization techniques to bound clique sizes while preserving the existence of a HIST.

\subsection{Parameterization by modular-width}
The modular-width of a graph is a structural parameter based on the concept of a module. A \emph{module} in a graph is a vertex subset such that every vertex outside the module is either adjacent to all vertices in the module or to none of them.
The \emph{modular decomposition} recursively partitions a graph into modules, forming a decomposition tree. Each node in the tree is classified as either a \emph{parallel node} (an independent set), a \textit{series node} (a clique), or a \emph{prime node} (neither).
The \emph{modular-width} is defined as the maximum size of a prime node in the modular decomposition tree.

Suppose that a graph $G$ is decomposed into modules $M_1, \ldots, M_k \subseteq V$. That is, for all $u, v \in M_i$ and $w \in V \setminus M_i$, if $\{u, w\} \in E$, then $\{v, w\} \in E$, and if $\{u, w\} \notin E$, then $\{v, w\} \notin E$.  
We assume that there exists a module consisting of a single vertex. If no such module exists, we can arbitrarily split one of the modules with at least two vertices into a singleton and the remainder. Although this increases the number of modules by one, as will be shown later, this has no impact on the computational complexity.

Concerning a HIST, we can show the following lemma. 
\begin{lemma}\label{lem:mod}
Suppose that the graph $G$ is decomposed into modules $M_0, \ldots, M_k \subseteq V$, where $|M_0| = 1$.  
If $G$ admits a HIST in which the vertex in $M_0$ is not a leaf, then there exists a HIST $T$ satisfying the following constraints:  
For each $i = 1, \ldots, k$, one of the following holds:
\begin{enumerate}
    \item All vertices in $M_i$ are leaves in $T$, i.e., $\forall v \in M_i: d_T(v) = 1$.
    \item There exists a vertex $u \in M_i$ such that $d_T(u) \geq 3$, and all other vertices in $M_i \setminus \{u\}$ are leaves in $T$, i.e., $\forall v \in M_i \setminus \{u\}: d_T(v) = 1$. In this case, one of the following holds:
    \begin{enumerate}
        \item[(a)] $d_T(u) = 3$ and the degree of $u$ within $T[M_i]$ is 1.
        \item[(b)] $d_T(u) \geq 3$ and the degree of $u$ within $T[M_i]$ is 0.
    \end{enumerate}
\end{enumerate}
\end{lemma}

\begin{proof}
It suffices to show that if the HIST of $G$ does not satisfy the conditions stated in the lemma, then it can be transformed into another HIST $T'$ that does satisfy them.

First, let the vertex in $M_0$ be the root $r$ (parent-child relationships are now defined). Suppose there exists a module $M_i$ ($i \geq 1$) that satisfies neither condition 1 nor condition 2. We divide the cases based on the number $l = |\{u \in M_i \mid d_T(u) \geq 3\}|$ of vertices in $M_i$ with degree at least 3.

\noindent\textbf{Case (i) $l = 1$:}  
Focus on the unique vertex $u \in M_i$ with $d_T(u) \geq 3$. Since $u$ is the only vertex of degree at least 3 in $M_i$, it must have a parent vertex $u'$ outside the module in $T$.  
Since $u$ does not satisfy condition 2, either of the following holds:  
(i-1) $d_T(u) > 3$ and $d_{T[M_i]}(u) = 1$, or  
(i-2) $d_{T[M_i]}(u) \geq 2$.

\textbf{(i-1):}  
Since $d_T(u) - d_{T[M_i]}(u) \geq 3$, we can reconnect all pendant vertices in $M_i \setminus \{u\}$ that were connected to $u$ directly to $u'$ while maintaining degree at least 3 for $u$.  
In the resulting tree, $M_i$ satisfies condition 2(b).

\textbf{(i-2):}  
At least two vertices in $M_i \setminus \{u\}$ are connected to $u$ as pendant vertices. We further divide into the following cases based on $d_T(u) - d_{T[M_i]}(u)$:
\begin{itemize}
    \item If $d_T(u) - d_{T[M_i]}(u) \geq 3$, we can apply the same argument as in (i-1) to transform $M_i$ to satisfy condition 2(b).
    \item If $d_T(u) - d_{T[M_i]}(u) = 1$, $u$ is connected only to its parent $u'$ and pendant vertices. By reconnecting all pendant vertices from $u$ to $u'$ without losing connectivity, $u$ itself becomes a pendant vertex of $u'$. In this way, we can transform $M_i$ to satisfy condition 1 while preserving the HIST structure.
    \item If $d_T(u) - d_{T[M_i]}(u) = 2$, since $d_T(u) \geq 4$, we can reconnect at least one pendant vertex from $M_i \setminus \{u\}$ to $u'$ while maintaining degree at least 3 for $u$.  
    In the resulting tree, if $d_T(u) = 4$, $M_i$ satisfies condition 2(a). If $d_T(u) > 4$, the situation reduces to subcase (i-1), which ultimately satisfies condition 2(b).
\end{itemize}

\noindent\textbf{Case (ii) $l \geq 2$:}  
This may hold only for $M_1,\ldots, M_k$. 
Select the vertex with degree at least 3 in $M_i$ that is closest to the root $r$ and denote it by $u$.  
Let $u'$ be its parent. 
Then, $d_T(u')\ge 3$ holds, 
since this is assumed for the root and true for all remaining non-leaves due to the HIST property of $T$. 
Since there is at least one other vertex of degree at least 3, choose such a vertex and denote it by $v$. Let $W$ be the set of children of $v$.

For each $w \in W \cap M_i$, reconnect the edge $\{v, w\}$ to $\{u', w\}$. For each $w \in W \setminus M_i$, reconnect $\{v, w\}$ to $\{u, w\}$. These reconnections maintain the tree structure, as the subtrees rooted at each $w$ can now be attached under $u'$ or $u$. Through this operation, the degrees of $u'$ and $u$ increase, while the degree of $v$ becomes 1. The degrees of all other vertices remain unchanged.
Since this operation reduces $l$ by 1, by repeatedly applying this process, we can eventually achieve $l = 1$, reducing the situation to case (i).

Thus, in all cases, the HIST can be transformed to satisfy the conditions stated in the lemma. \qed
\end{proof}

Based on this lemma, we further see how a HIST may form in the \emph{quotient graph}.
Suppose that $G$ is decomposed into a set of modules $M_0, \ldots, M_k$ such that $|M_0| = 1$. If $G$ admits a HIST $T$, then from Lemma \ref{lem:mod}, each module contains at most one vertex of degree at least 3 in $T$. We refer to such a vertex as the \emph{representative vertex} of the module. Consider the quotient graph $H = (\mathcal{M}, \mathcal{E})$ where $\mathcal{M} = \{M_i \mid i = 0, \ldots, k\}$ represents the modules of $G$.  
Let $\mathcal{M'}$ denote the set of modules that contain a representative vertex.  
The subgraph defined by the connections between representative vertices forms a tree.

Note that, in this structure, the degrees of the representative vertices themselves do not reach three if we consider only the connections between representatives.  
The HIST is completed by connecting pendant vertices to the representative vertices.  
Furthermore, within each module, it suffices to assume that exactly one vertex is connected as a pendant vertex.

To formalize these conditions, we partition the set $\mathcal{M'}$ in the quotient graph $H = (\mathcal{M}, \mathcal{E})$ (where $\mathcal{M} = \{M_i \mid i = 0, \dots, k\}$) as follows.  
For a spanning tree $T$ of the subgraph $H[\mathcal{M'}]$, we divide $\mathcal{M'}$ into:
\begin{align*}
\mathcal{M}_1 & := \{M \in \mathcal{M'} \mid d_T(M) = 1\}, \\    
\mathcal{M}_2 & := \{M \in \mathcal{M'} \mid d_T(M) = 2\}, \\
\mathcal{M}_3 & := \{M \in \mathcal{M'} \mid d_T(M) \geq 3\}.
\end{align*}
Let $\mathcal{M}^{\bot}$ denote the set of modules that form independent sets.

The following conditions must hold:
\begin{align}\label{eq:assign1}
    \forall M_i \in \mathcal{M}_1: & \sum_{M_j : \{M_i, M_j\} \in \mathcal{E}} x_{ji} + x_{ii} \geq 2 \\ \label{eq:assign2}
    \forall M_i \in \mathcal{M}_2: & \sum_{M_j : \{M_i, M_j\} \in \mathcal{E}} x_{ji} + x_{ii} \geq 1 \\ \label{eq:assign3}
    \forall M_i \in \mathcal{M'}: & \sum_{M_j : \{M_i, M_j\} \in \mathcal{E}} x_{ij} + x_{ii} = |M_i| - 1 \\ \label{eq:assign4}
    \forall M_i \in \mathcal{M} \setminus \mathcal{M'}: & \sum_{M_j : \{M_i, M_j\} \in \mathcal{E}} x_{ij} = |M_i| \\ \label{eq:assign5}
    \forall M_i \in \mathcal{M'}: & x_{ii} \in \begin{cases}
      \{0\}  & \text{if } M_i \in \mathcal{M}^{\bot} \\
      \{0,1\} & \text{otherwise}
    \end{cases}.
\end{align}

Here, $x_{ij}$ denotes the number of vertices in module $M_i$ that are connected as pendant vertices to the representative vertex of module $M_j$.  
Also, $x_{ii}$ represents the number of pendant vertices connected from within module $M_i$ to its own representative vertex. From Lemma \ref{lem:mod}, we know that $x_{ii}$ is at most one.  

Note that if the module itself is not an independent set, such an assignment is always feasible (if necessary). This feasibility is precisely expressed by condition (\ref{eq:assign5}).

\medskip

This assignment problem determines whether a HIST of $G$ can be constructed based on the spanning tree $T$ of $H[\mathcal{M}']$.  
In the desired HIST, the representative vertices of $\mathcal{M}'$ must have degree 3.  
However, the representative vertices in the modules of $\mathcal{M}_1$ and $\mathcal{M}_2$ can only attain degrees of 1 or 2 by using the edges present in $T$ alone.  
Therefore, this formulation verifies whether it is possible to assign pendant vertices to achieve the required degrees appropriately.

The right-hand sides of constraints (\ref{eq:assign1}) and (\ref{eq:assign2}) represent the required number of pendant vertices for each module, while the left-hand sides specify from which modules these pendant vertices can be assigned.
Constraints (\ref{eq:assign3}) and (\ref{eq:assign4}) ensure that the total number of assigned pendant vertices equals $|M_i| - 1$ (excluding the representative vertex) or $|M_i|$ (when all vertices are assigned as pendants), respectively.
If this assignment problem has a feasible solution, then by choosing the representative vertices as the internal vertices of the tree $T$ and connecting the remaining vertices according to the assignment solution as pendant vertices, we can construct a HIST of $G$.

It should be noted that the above discussion focuses solely on the fact that $T$ is a spanning tree of $H$ and on the degrees within each module. From this observation, the following result holds.

\begin{lemma}\label{lem:mod2}
Suppose that the graph $G$ is decomposed into a set of modules $\mathcal{M} = \{M_i \subseteq V \mid i = 0, \ldots, k\}$, where $|M_0| = 1$.  
If there exists a subset of modules $\mathcal{M}' \subseteq \mathcal{M}$ and a partition $(\mathcal{M}_1, \mathcal{M}_2, \mathcal{M}_3)$ of $\mathcal{M}'$ satisfying the following two conditions, then $G$ admits a HIST:
\begin{enumerate}
    \item There exists a spanning tree $T$ of the subgraph $H[\mathcal{M}']$ of the quotient graph $H = (\mathcal{M}, \mathcal{E})$ of $G$ such that $\mathcal{M}_1 := \{M \in \mathcal{M}' \mid d_T(M) = 1\}$, $\mathcal{M}_2 := \{M \in \mathcal{M}' \mid d_T(M) = 2\}$, and $\mathcal{M}_3 := \{M \in \mathcal{M}' \mid d_T(M) \geq 3\}$.
    \item The assignment problem defined by constraints (\ref{eq:assign1})--(\ref{eq:assign5}) has a feasible solution.
\end{enumerate}
\end{lemma}

Based on Lemma \ref{lem:mod2}, we can determine the existence of a HIST in $G$ by exhaustively enumerating all subsets $\mathcal{M}' \subseteq \mathcal{M}$, considering all possible partitions $(\mathcal{M}_1, \mathcal{M}_2, \mathcal{M}_3)$, verifying the existence of a spanning tree $T$ in $H[\mathcal{M}']$ that satisfies condition (1), and solving the corresponding assignment problem to check condition (2).  

The former can be verified using a slight modification of \Cref{alg:3}. The latter can be solved in polynomial time, which leads to the following theorem.

\begin{theorem}\label{thm:mw}
Given an $n$-vertex graph with modular-width $k$, the existence of a HIST can be determined in $O^*(4^k)$ time.
\end{theorem}

\subsection{Parameterization by treewidth}

In this section, we demonstrate the fixed-parameter tractability of computing a HIST for treewidth. To this end, it is sufficient to give an appropriate MSO$_2$ formulation. 

It is well-known that verifying an edge subset $F$ is a tree can be expressed as an MSO$_2$ formula $\tree(F)$~\cite{dimacs/Courcelle96}. Then, a constant-length MSO$_2$ formula verifying $F$ is a HIST can be expressed as follows:    
\begin{align*}
    \varphi(F) &= \tree(F)
    \land \big(\forall v\in V.\ 
    \exists e_1\in F.\ \inc(v,e_1)\\
    & \hspace{2.3cm} \land 
    (\exists e_1,e_2\in F.\ (e_1\neq e_2)\land \inc(v,e_1)\land \inc(v,e_2)))\\
    &\hspace{2.5cm} \implies \exists e_1,e_2,e_3\in F.\   (e_1\neq e_2)\land  (e_2\neq e_3)\land  (e_3\neq e_1)\\
    &\hspace{4cm}\land \inc(v,e_1)\land \inc(v,e_2)\land \inc(v,e_3) \big)
\end{align*}
The first line guarantees that $F$ is a tree and $F$ spans $V$, i.e., $F$ is a spanning tree in $G$. 
The second, third, and fourth lines mean that if a vertex is incident to two edges $e_1,e_2$, then it is incident to three edges $e_1,e_2,e_3$. Since every vertex has at least one edge in $F$ by the first line, this guarantees that every vertex has either exactly one edge or at least three edges. 
Therefore, $\varphi(F)$ verifies whether $F$ is a HIST.
Since the length of $\varphi$ is constant, by Courcelle's theorem~\cite{iandc/Courcelle90,jal/ArnborgLS91}, finding a HIST is fixed-parameter tractable when parameterized by treewidth.

\begin{theorem}\label{thm:treewidth}
    Finding a HIST is fixed-parameter tractable when parameterized by treewidth.
\end{theorem}

For a more concrete FPT algorithm, we can design a dynamic programming algorithm on a nice tree decomposition. 
It is not difficult to show that we can get time bound $2^{O(\tw \log \tw)}n^{O(1)}$, where $\tw$ is the treewidth of an input graph, though we omit the details.




\subsection{Parameterization by cluster deletion number}\label{sec:cvd} 
A graph in which each connected component is a complete graph (clique) is called a \emph{cluster graph}.  
The \emph{cluster vertex deletion number} of a graph $G = (V, E)$ is the minimum number of vertices that need to be deleted to make the graph a cluster graph.  
Formally, the cluster vertex deletion number $\cvd(G)$ is defined as follows:
\[
\cvd(G) = \min_{S \subseteq V} \{ |S| \mid G \setminus S \text{ is a cluster graph} \}.
\]

Let $S$ be a vertex subset such that $G \setminus S$ is a cluster graph. Suppose that $G \setminus S$ consists of a set of connected cliques $\mathcal{C}$.

\begin{lemma}\label{lem:cvd}
Let $G$, $S$, and $\mathcal{C}$ be as defined above, and suppose that $G$ admits a HIST.  
Then, there exists a HIST $T$ of $G$ satisfying the following:  
For each clique $C \in \mathcal{C}$ and for each partition of $C$ into true twin vertex sets $(M_1, \ldots, M_l)$, each $M_i$ contains at most one vertex of degree at least 3 in $T$; that is,
\[
|\{ u \in M_i \mid d_T(u) \geq 3 \}| \leq 1.
\]
\end{lemma}

\begin{proof}
We show that if a HIST $T$ of $G$ does not satisfy this condition, it can be transformed into one that does.  
Suppose that in some $C \in \mathcal{C}$ and some $M_i$, there exist at least two vertices of degree at least 3 in $T$.  
Let $u$ and $v$ be such vertices with $d_T(u) \geq 3$ and $d_T(v) \geq 3$.  
In the original graph $G$, $u$ and $v$ are true twins.

We consider $T$ as a rooted tree with a fixed ordering and assume without loss of generality that $u < v$ in this ordering.  
By reattaching all children of $v$ to $u$, the degree of $v$ becomes 1, the degree of $u$ increases by at least two, and the degrees of all other vertices remain unchanged.  
The tree remains connected after this operation.  
The new tree $T'$ is thus a HIST with one fewer vertex of degree at least 3 than $T$.

As long as there exists some $M_i$ containing two or more vertices of degree at least 3, we can repeatedly apply this operation.  
Eventually, we obtain a HIST that satisfies the desired condition. \qed
\end{proof}

From \Cref{lem:cvd}, in each clique $C$, for a module $M \subseteq C$, if $|M| \geq l + |S| + 3$, we can safely delete vertices from $M$ until $|M| = l + |S| + 1$.

\begin{claim}
Let $D \subseteq M \subseteq C$ be the set of deleted vertices, and let $G' = G \setminus D$ denote the graph after deleting $D$. Then, $G$ admits a HIST if and only if $G'$ admits a HIST. 
\end{claim}

\begin{claimproof}
Suppose that $G$ admits a HIST. From \Cref{lem:cvd}, there exists a HIST $T$ in which each module contains at most one vertex of degree at least 3.  
Without loss of generality, we can assume that the degree-3 vertex $v \in M$ is not included in $D$.

In $T$, each vertex in $D$ is adjacent as a leaf to at most $l + |S|$ vertices of degree at least 3.  
Consider the tree $T \setminus D$ in the graph $G' = G \setminus D$.  
If there exists a vertex in $G \setminus D$ that does not satisfy the degree condition of a HIST in $T \setminus D$, such a vertex must either be a degree-3 vertex within $C$ in $T$ or a vertex in $S$.  
The number of such vertices is at most $l + |S|$, and all are adjacent to all vertices in $M$ with degree 2.
Note that the vertices in $M$ are true twins and that the vertices in $D$ are leaves in $T$.  
Since $|M \setminus (D \cup \{v\})| = l + |S|$, we can reassign one leaf from $M \setminus D$ to each of these vertices, increasing their degrees to at least 3 and satisfying the degree conditions of a HIST. Therefore, $G'$ admits a HIST.

Conversely, if $G'$ admits a HIST $T'$, we can attach the vertices in $D$ as leaves to any vertices in $V \setminus D$ in $T'$.  
Since $|D| \geq 2$, the resulting tree is a HIST in $G$.
\end{claimproof}

Since the number of modules in each clique $C$ of $G \setminus S$ is at most $2^{\cvd}$, we obtain the following.

\begin{lemma}\label{lem:cvd2}
By applying the above reduction, the size of the maximum clique in $G$ can be reduced to at most $2^{\cvd}(2^{\cvd} + 3 + \cvd) + \cvd$.
\end{lemma}

From \Cref{lem:cvd2}, it follows that in graphs where the cluster vertex deletion number $\cvd$ is bounded, the size of the graph can be reduced while preserving the existence of a HIST, such that its treewidth is at most $2^{\cvd+1}(2^{\cvd+1} + 2 + 2\cvd) + \cvd$.  
Therefore, by \Cref{thm:treewidth}, the problem is fixed-parameter tractable (FPT) with respect to $\cvd$.

\begin{theorem}\label{thm:cvd}
    Finding a HIST is fixed-parameter tractable when parameterized by cluster deletion number.
\end{theorem}
%
%
%
 \bibliographystyle{splncs04}
 \bibliography{HIST}
\end{document}